\documentclass[12pt]{article}

\newcommand{\commentout}[1]{}

\usepackage{amsmath, amsthm, amssymb,amscd}
\usepackage{fullpage, paralist,enumerate}
\usepackage{verbatim}
\usepackage[all]{xy}
\usepackage[parfill]{parskip}
\usepackage{graphicx}

\usepackage{afterpage}

\usepackage{amsmath}
\usepackage{amsfonts}
\usepackage{mathcomp}
\usepackage{textcomp}
\usepackage{amssymb}
\usepackage{latexsym}
\usepackage{graphicx}
\usepackage{courier}
\usepackage[english]{babel}
\usepackage{mathbbol}
\usepackage{multirow}
\usepackage{latexsym}
\usepackage{epsfig}
\usepackage{subfigure}
\usepackage{dcolumn}
\usepackage{bm}
\usepackage{colordvi}
\usepackage{color}
\usepackage{booktabs}
\usepackage{stmaryrd}
\usepackage{cite}
\usepackage[linesnumbered,commentsnumbered,ruled]{algorithm2e}
\usepackage{epstopdf}
\input xy
\xyoption{all}

\newtheorem{thm}{Theorem}[section]
\newtheorem{prop}[thm]{Proposition}

\newtheorem{rmk}[thm]{Remark}

\newcommand{\nwc}{\newcommand*}

\nwc{\ben}{\begin{equation*}}
\nwc{\bea}{\begin{eqnarray}}
\nwc{\beq}{\begin{eqnarray}}
\nwc{\bean}{\begin{eqnarray*}}
\nwc{\beqn}{\begin{eqnarray*}}
\nwc{\beqast}{\begin{eqnarray*}}

\nwc{\eal}{\end{align}}
\nwc{\een}{\end{equation*}}
\nwc{\eea}{\end{eqnarray}}
\nwc{\eeq}{\end{eqnarray}}
\nwc{\eean}{\end{eqnarray*}}
\nwc{\eeqn}{\end{eqnarray*}}

\CompileMatrices

\theoremstyle{remark}

\nwc{\nn}{\nonumber}
\nwc{\mb}{\mathbf}
\nwc{\ml}{\mathcal}

\newcommand{\lt}{\left}
\newcommand{\rt}{\right}

\nwc{\vep}{\varepsilon}
\nwc{\ep}{\epsilon}
\nwc{\vrho}{\varrho}
\nwc{\orho}{\bar\varrho}
\nwc{\vpsi}{\varpsi}
\nwc{\lamb}{\lambda}
\nwc{\om}{\omega}
\nwc{\Om}{\Omega}
\nwc{\al}{\alpha}

\nwc{\IA}{\mathbb{A}} 
\nwc{\bi}{\mathbf i}
\nwc{\bo}{\mathbf o}
\nwc{\IB}{\mathbb{B}}
\nwc{\IC}{\mathbb{C}} 
\nwc{\ID}{\mathbb{D}} 
\nwc{\IM}{\mathbb{M}} 
\nwc{\IP}{\mathbb{P}} 
\nwc{\II}{\mathbb{I}} 
\nwc{\IE}{\mathbb{E}} 
\nwc{\IF}{\mathbb{F}} 
\nwc{\IG}{\mathbb{G}} 
\nwc{\IN}{\mathbb{N}} 
\nwc{\IQ}{\mathbb{Q}} 
\nwc{\IR}{\mathbb{R}} 
\nwc{\IT}{\mathbb{T}} 
\nwc{\IZ}{\mathbb{Z}} 

\nwc{\cE}{{\ml E}}
\nwc{\cP}{{\ml P}}
\nwc{\cQ}{{\ml Q}}
\nwc{\cL}{{\ml L}}
\nwc{\cX}{{\ml X}}
\nwc{\cW}{{\ml W}}
\nwc{\cZ}{{\ml Z}}
\nwc{\cR}{{\ml R}}
\nwc{\cV}{{\ml V}}
\nwc{\cT}{{\ml T}}
\nwc{\crV}{{\ml L}_{(\delta,\rho)}}
\nwc{\cC}{{\ml C}}
\nwc{\cO}{{\ml O}}
\nwc{\cA}{{\ml A}}
\nwc{\cK}{{\ml K}}
\nwc{\cB}{{\ml B}}
\nwc{\cD}{{\ml D}}
\nwc{\cF}{{\ml F}}
\nwc{\cS}{{\ml S}}
\nwc{\cM}{{\ml M}}
\nwc{\cG}{{\ml G}}
\nwc{\cH}{{\ml H}}
\nwc{\bk}{{\mb k}}
\nwc{\bn}{{\mb n}}
\nwc{\bz}{\mb z}
\nwc{\by}{\mathbf{h}}
\nwc{\bZ}{\mathbf{Z}}
\nwc{\bF}{\mathbf{F}}
\nwc{\bE}{\mathbf{E}}
\nwc{\bV}{\mathbf{V}}
\nwc{\bY}{\mathbf Y}
\nwc{\br}{\mb r}
\nwc{\pft}{\cF^{-1}_2}
\nwc{\bU}{{\mb U}}
\nwc{\bG}{{\mb G}}
\nwc{\bg}{\mathbf{g}}
\nwc{\mbf}{\mathbf{f}}
\nwc{\mbe}{\mathbf{e}}
\nwc{\be}{\mathbf{e}}
\nwc{\ind}{\operatorname{I}}
\nwc{\mbx}{\mathbf{f}}
\nwc{\bb}{\mathbf{g}}
\nwc{\xmax}{f_{\rm max}}
\nwc{\xmin}{f_{\rm min}}
\nwc{\suppx}{\hbox{\rm supp} (\mbf)}
\nwc{\cI}{\IZ^2_N}
\nwc{\chis}{{\chi^{\rm s}}}
\nwc{\chii}{{\chi^{\rm i}}}
\nwc{\pdfi}{{f^{\rm i}}}
\nwc{\pdfs}{{f^{\rm s}}}
\nwc{\pdfii}{{f_1^{\rm i}}}
\nwc{\pdfsi}{{f_1^{\rm s}}}
\nwc{\thetatil}{{\tilde\theta}}
\nwc{\red}{\color{red}}
\nwc{\blue}{\color{blue}}

\nwc{\prox}{\hbox{prox}}
\nwc{\diag}{\hbox{\rm diag}}
\nwc{\supp}{{\hbox{\rm supp}}}

\nwc{\sloc}{J_{\rm f}}
\nwc{\bu}{\xi}
\nwc{\bv}{\eta}
\nwc{\cU}{\mathcal{U}}
\nwc{\cN}{\mathcal{N}}
\nwc{\bN}{\mathbf{N}}
\nwc{\mbm}{\mathbf{m}}
\nwc{\bw}{\mathbf{w}}
\nwc{\im}{i}
\nwc{\bom}{\mathbf{w}}
\nwc{\bt}{\mathbf{t}}
\nwc{\z}{y}
\nwc{\cY}{\mathcal{Y}}
\nwc{\bM}{\mathbf{M}}
\nwc{\half}{{1\over 2}}
\nwc{\Sf}{S_{\rm f}}
\nwc{\Jf}{J_{\rm f}}
\nwc{\nul}{\hbox{\rm null}_\IR}
\nwc{\spanR}{\hbox{\rm span}_\IR}
\nwc{\Arg}{\hbox{\rm Arg~}}
\nwc{\fdr}{S_{\rm f}}
\nwc{\phase}[1]{\exp\lt[i\measured #1\rt]}
\nwc{\xnul}{x_{\rm null}}

\begin{document}

\title{
Phase Retrieval by Linear Algebra}

\author{Pengwen Chen
\thanks{
Department of Applied Mathematics, National Chung Hsing University, Taichung  402, Taiwan. 
}
\and Albert Fannjiang
\thanks{Department of Mathematics, University of California, Davis, CA 95616, USA. 
{\tt fannjiang@math.ucdavis.edu.} 
}
\and Gi-Ren Liu
\thanks{Department of Mathematics, University of California, Davis, CA 95616, USA}
}

\maketitle

\begin{abstract}  
 The null vector method, based on a simple linear algebraic concept,  is proposed as a solution to the phase retrieval problem. 
 In the case with complex Gaussian random measurement matrices, a non-asymptotic  error bound  is derived, yielding an asymptotic regime of accurate approximation comparable to that for the spectral vector method. 
 
 \end{abstract}
\commentout{

\pagestyle{myheadings}
}

\maketitle

\section{Introduction}

We consider the following phase retrieval problem:
Let $A=[a_i]$ be a $n\times N$ random matrix with independently and identically distributed
entries in $N(0,1)+iN(0,1)$, i.e. circularly symmetric complex Gaussian random variables.
Let $x_0\in \IC^n$ and $y=A^* x_0$. 
Suppose we are given $A$ and $b:=|y|$ where  $|y|$ denote  the vector such that $|y|(j)=|y(j)|,\forall j$. The aim of  phase retrieval is to  find $x_0$.

Clearly this is a nonlinear inversion problem. Simple dimension count shows that, for the solution to be unique in general, the number of (nonnegative) data $N$ needs to be at least  twice the number $n$ of unknown (complex) components. 
There are many approaches to phase retrieval, the most efficient and effective, especially when the problem size is large,  being
fixed point algorithms (see \cite{Mar07, ER, DR-phasing,NJS} and references therein) and non-convex optimization
methods \cite{CLS2,truncatedWF}. Phase retrieval has a wide range of applications,
see \cite{Sh} for a recent survey.

 The following observation motivates  our current approach:  
Let $I$ be a subset of $\{1,\cdots,N\}$  and $I_c$ its complement such that $b(i)\leq b(j)$ for all $i\in I, j\in I_c$.  In other words, $\{b(i): i\in I\}$ are the ``weaker" signals and
$\{b(j): j\in I_c\}$ the ``stronger" signals. 
 Let $|I|$ be the cardinality of the set $I$.
 Since $b(i)=|a_i^*x_0|, i\in I, $ are small, $\{ a_i\}_{i\in I}$ is  a set of sensing vectors  nearly orthogonal  to $x_0$. 
 Denote  the sub-column  matrices  consisting of $\{a_i\}_{i\in I} $ and $\{a_j\}_{j\in I_c}$  by $A_I$ and $A_{I_c}$, respectively. 
Define the null vector by the singular vector for the least singular value of $A_I$:  \[
x_{\rm null}:=\hbox{\rm arg}\min\lt\{\|A^*_I x\|^2: x\in \IC^n, {\|x\|=\|x_0\|}\rt\} 
 \]
 which can be computed by  purely linear algebraic methods.

The goal of the paper is to establish a regime where $x_{\rm null}$ is
an accurate approximation to $x_0$.

 \commentout{ Hence we should
 consider the global phase adjustment for a given null vector $x_{\rm null}$
 \beq
\nn
 \min_{\alpha \in \IC,\; |\alpha|=1 }\| \alpha x_{\rm null}-x_0\|^2=2\|x_0\|^2-2\max_{|\alpha|=1} \Re(x_0^*\alpha\xnul).
 \eeq
In what follows, we assume $\xnul$ to be optimally adjusted so  that 
\beq
 \label{52'}\|x_{\rm null}-x_0\|^2=2\|x_0\|^2-2 |x_0^*\xnul| 
\eeq
Also, in many imaging problems, the norm of the true object, like the constant phase factor, is either  recoverable by other prior information or irrelevant to the quality of reconstruction.}

 \commentout{
 Define the dual vector 
  \beq
x_{\rm dual}:= \hbox{\rm arg}\max\lt\{ \|A_{I_c}^*  x\|^2: x\in \cX, {\|x\|=\|x_0\|}\rt\} \eeq
 whose phase factor is optimally adjusted as $x_{\rm null}$.  
}

\commentout{
 \subsection{Isometric $A^*$}
 
For isometric  $A^*$, 
 \beq
 \label{2.3}
x_{\rm null}:=\hbox{\rm arg}\min\lt\{\sum_{i\in I} \|a_i^* x\|^2: x\in \cX, {\|x\|=\|b\|}\rt\}.
 \eeq
We have  
 \[
 \|A_I^*  x\|^2+\|A_{I_c}^*  x\|^2=\|x\|^2
\]
and hence 
\beq
\label{51'} x_{\rm null}=x_{\rm dual},
\eeq
i.e. the null vector is self-dual in the case of isometric $A^*$. 
Eq.  \eqref{51'}  can be used to  construct
the null vector from  $A_{I_c}A^*_{I_c}$ by the power method. 

Let $\mathbf{1}_{c}$ be the characteristic function of the complementary index $I_c$
with $|I_c|=\gamma N$.   The default choice for $\gamma$ is the  median value $\gamma=0.5$.
\begin{algorithm}
\SetKwFunction{Round}{Round}

\textbf{Random initialization:} $x_{1}=x_{\rm rand}$
\\
\textbf{Loop:}\\
\For{$k=1:k_{\textup{max}}-1$}
{
$x'_{k}\leftarrow A (\mathbf{1}_{c}\odot A^*x_k)$;\\
$x_{k+1}\leftarrow [x_k^{'}]_\cX/\|[x_k^{'}]_\cX\|$
}
{\bf Output:} {$\hat x_{\textup{dual}}=x_{k_{\textup{max}}}$.}
\caption{\textbf{The  null initialization}}
\label{null-algorithm}
\end{algorithm}

{
For isometric $A^*$, it is natural to define
\beq
\label{2.4}
x_{\rm null}=\alpha\|b\| \cdot \hat x_{\rm dual},\quad \alpha= {\hat x_{\rm dual}^*x_0\over
| \hat x_{\rm dual}^*x_0|}\eeq
where $\hat x_{\rm dual}$ is the output of Algorithm 1. 
As shown in Section \ref{sec:num} (Fig. \ref{fig:noise0}), 
the null vector is remarkably stable with respect to
noise in $b$. }

%


\subsection{Non-isometric $A^*$} When $A^*$ is non-isometric such as the standard Gaussian random matrix (see below),  the power method is still applicable with  the following modification.

For a full rank $A$, let $A^*=QR$ be  the QR-decomposition of $A^*$ where
$Q$ is isometric and $R$ is a full-rank, upper-triangular square matrix. Let $z=Rx$, $z_0=Rx_0$ and $z_{\rm null}=Rx_{\rm null}$. Clearly, $z_{\rm null}$ is the null vector for the isometric phase retrieval
problem $b=|Q z|$ in the sense of \eqref{2.3}. 

Let $I$ and $I_c$ be the index sets as above.
Let 
\beq
\hat z=\hbox{\rm arg}\max_{\|z\|=1} \|Q_{I_c}z\|.
\eeq
{
Then
\[
x_{\rm null}=\alpha\beta R^{-1}\hat  z
\]
where $\alpha$ is the optimal phase factor and 
\[ \beta={ \|x_0\|\over \|R^{-1} \hat z\|}
\]
may be an  unknown parameter in the non-isometric case.
As pointed out above, when $x_{\rm null}$ with an arbitrary parameter $\beta$ is used as initialization of phase retrieval, the first iteration of AP would recover the true value of $\beta$ as AP is totally independent of any real constant factor. }

}
\section{Approximation  theorem}\label{sec:Gaussian}
Note that both $x_{\rm null}$ and the phase retrieval solution  is at best uniquely defined
up to a global phase factor. So 
we use 
the following error  metric 
 \beq
 \label{52''}  \|x_0x_0^* -x_{\rm null}x_{\rm null}^*\|^2 
=2\|x_0\|^4- 2|x_0^*\xnul|^2
\eeq 
which has the advantage of being independent of the global phase factor. 

The following theorem is our main result. 
\begin{thm} \label{Gaussian}
Suppose 
\beq
\label{53'}
\sigma:={|I|\over N}<1,\quad \nu={n\over |I|}<1.
\eeq
Then for any $\ep\in (0,1),\delta>0$ and $t\in (0, \nu^{-1/2}-1)$ 
 the following error bound 
\beq\label{error}
\|x_0x_0^* -x_{\rm null}x_{\rm null}^*\|^2
&\le& \lt( \left(\frac{2+t}{1-\epsilon} \right) \sigma+\ep \lt(-2\ln (1-\sigma) +\delta\rt)\rt){ 2\|x_0\|^4\over \left(1-(1+t)\sqrt{\nu}\right)^{2}}
 \eeq
holds with probability at least \beq
\label{prob}
&& 1-2\exp\left(-{N}{\delta^2 e^{-\delta}|1-\sigma|^2/2} \right)-
\exp(-{2}  \lfloor |I| \epsilon \rfloor^2/N )-Q
 \eeq
where $Q$ has the asymptotic upper bound 
 \beq\label{a.6}
 2 \exp\lt\{-c\min \lt[{e^2t^2 \over 16} \lt(\ln \sigma^{-1}\rt)^2 {|I|^2/N},~{et\over 4}|I|\ln\sigma^{-1}\rt]\rt\}, \quad\sigma \ll 1,\eeq
with  an absolute constant $c$. 
\end{thm}
\begin{rmk}
To unpack the implications of Theorem \ref{Gaussian}, consider the following asymptotic:
With $\ep$ and $ t $  fixed,  let 
\[
n\gg 1,\quad  \sigma={|I|\over N}\ll 1,\quad {|I|^2\over N}\gg 1, \quad\nu={n\over |I|}<1.
\]
We have
\beq
\label{5.8.1}
\|x_0x_0^* -x_{\rm null}x_{\rm null}^*\|^2\le c_0 \sigma\|x_0\|^4
\eeq
with probability at least
\[
1-c_1 e^{-c_2 n}-  c_3 \exp\lt\{-c_4\lt(\ln\sigma^{-1}\rt)^2 {|I|^2/ N} \rt\}
\]
  for moderate constants  $c_0, c_1, c_2,c_3,c_4$.
   \label{rmk5.2}
\end{rmk}
  The proof of Theorem  \ref{Gaussian} is given in the next section.

The spectral vector method \cite{NJS,CLS2,truncatedWF} is another
  linear algebraic method and uses the leading singular
  vector $x_{\rm spec}$   of $B^*=\diag{[b]}A^*$   to approximate $x_0$
  where      \[
x_{\rm spec}:=\hbox{\rm arg}\max\lt\{\|B^*x\|^2: x\in \IC^n, {\|x\|=\|x_0\|}\rt\}. 
 \]
 The spectral vector method has a comparable performance guarantee to \eqref{5.8.1}
 which vanishes as $\sigma\to 0$, with probability close to 1 exponentially in $n$. 
 
  \commentout{
  If we set $N=Cn$, 
  then \eqref{5.8.1} becomes
\beq
\label{5.8.2}
\|x_0x_0^* -x_{\rm null}x_{\rm null}^*\|^2\le {c_0\over C\nu} \|x_0\|^4, 
\eeq
which is arbitrarily small with a sufficiently large $C\nu$, with probability close to 1 exponentially in $n$. 
}

In practice, however, the null vector method significantly outperforms the spectral vector method
in terms of accuracy and noise stability when $b$ is contaminated with noise \cite{ER}. 


The drawback with both approaches is that the error metric vanishes only with infinitely
many data, 
$N\to \infty$. For a finite data set, the null vector is best to be deployed in conjunction with 
a fast (locally) convergent fixed point algorithm such as alternating projection \cite{ER}
or the Douglas-Rachford algorithm \cite{DR-phasing}. 

\commentout{
In comparison, the performance guarantee for the spectral initialization (\cite{CLS2}, Theorem 3.3)  assumes 
$N=O(n\log n)$ for the same level of accuracy guarantee  with a success probability
less than $1-8/n^2$. On the other hand, the performance guarantee for the truncated
spectral vector  is comparable to Theorem \ref{Gaussian} in the sense that error bound like \eqref{5.8.2} holds true for the truncated spectral vector with $N=Cn$ and probability
exponentially close to 1 (\cite{truncatedWF}, Proposition 3). 
}

 \section{Proof of Theorem ~\ref{Gaussian}}

  The proof is based on the following two propositions.
 \begin{prop} \label{prop5.1}  
 There exists $x_\bot \in \IC^n$ with $x_\bot^* x_0=0$ and $\|x_\bot\|=\|x_0\|=1$  such that 
  \begin{eqnarray}\label{closeness1}
    \frac{1}{4}\|x_0x_0^* -x_{\rm null} x_{\rm null}^* \|^2 &\leq &  {\|b_I\|^2\over \|A_I^* x_\bot \|^2}. \end{eqnarray}

\commentout{
In terms of relative errors,
\[
\min_{\alpha \in \IC,\; |\alpha|=1 }\| x_{\rm null}-\alpha x_0\|^2=2(1-\beta)\le 4 (\frac{1-\beta^2}{2-\beta^2})\le 4\frac{\|b_I\|^2}{\|A^*_I x_\bot\|^2}.
\]
}
 \end{prop}
 \begin{proof}
Since $\xnul$ is optimally phase-adjusted, we have
\beq
\label{52}
\beta: =x_0^*\xnul\ge 0
\eeq
and \beq
\label{a.2}
 x_0=\beta x_{\rm null}+\sqrt{1-\beta^2}\, z 
 \eeq
for some unit vector $z^*x_{\rm null}=0$.
Then \beq
\label{571} x_\bot :=-(1-\beta^2)^{1/2} x_{\rm null}+\beta z 
\eeq
 is a unit vector  satisfying $x_0^*x_\perp=0$. Since $\xnul$ is a singular vector and $z$ belongs in
 another  singular subspace, we have 
  \beqn
\|A_I^* x_0\|^2&=&\beta^2\|A_I^* x_{\rm null}\|^2+(1-\beta^2)\|A_I^* z\|^2, \\
 \|A_I^* x_\bot \|^2&=&(1-\beta^2)\|A_I^*x_{\rm null}\|^2+\beta^2\|A_I^* z\|^2 
  \eeqn
from which it follows that 
  \beq\label{57}
&&(2-\beta^2)\|A_I^* x_0\|^2-(1-\beta^2)\|A_I^* x_\bot \|^2\\
 &=&\|A_I^* x_{\rm null}\|^2+2(1-\beta^2)^2\left(
 \|A_I^* z\|^2-\|A_I^* x_{\rm null}\|^2\right)\ge 0.\nn
 \eeq
 \commentout{
and 
\beq
\label{56'}
\|A_I^* x_\bot \|^2 - \|A_I^* x_0\|^2=
(2\beta^2-1)  \left(
 \|A_I^* z\|^2-\|A_I^* x_{\rm null}\|^2\right)\geq 0
 \eeq
 by \eqref{52}. 
 }
By  \eqref{57}, \eqref{52''} and $\|b_I\|= \|A_I^*x_0\|$, we also have   \begin{eqnarray}
 &&  \frac{\|b_I\|^2}{  \|A_I^* x_\bot \|^2} 
 \ge {1-\beta^2\over 2-\beta^2} \ge \half(1-\beta^2)=  \frac{1}{4}\|x_0x_0^* -x_{\rm null} x_{\rm null}^* \|^2.
 \eeq

 \end{proof}

   \begin{prop}\label{BnormA}
Let  $A\in \IC^{n\times N}$ be  an i.i.d.  complex standard Gaussian random matrix.
Then for any $\ep>0, \delta>0, t>0$
 \[ \|b_I\|^2\le |I| \lt( \left(\frac{2+t}{1-\epsilon} \right) \frac{|I|}{N}+\ep \lt(-2\ln \lt(1-{|I|\over N}\rt)+\delta\rt)\rt)
\]
with probability at least 
\beq
\label{prob'}
&& 1 -2\exp\left(-{N}{\delta^2 e^{-\delta}|1-\sigma|^2/2} \right)-
2\exp\lt(-{2\ep^2 |1-\sigma|^2} \sigma^2 N\rt)-Q \nn
 \eeq
 where $Q$ has the asymptotic upper bound
 \beqn
2 \exp\lt\{-c\min \lt[{e^2t^2 \over 16} {|I|^2\over N} \lt(\ln \sigma^{-1}\rt)^2,~{et\over 4}|I|\ln\sigma^{-1}\rt]\rt\},\quad  \sigma:={|I|\over N}\ll 1. \eeqn\end{prop}
 The proof of Proposition \ref{BnormA} is given in the next section.  \\

Now we turn to the proof of Theorem  \ref{Gaussian}.

Without loss of the generality we may assume $\|x_0\|=1$. Otherwise, we replace $x_0, x_{null}$ by $x_0/\|x_0\|$ and $x_{\rm null}/\|x_0\|$, respectively. 
Let $ Q=[Q_{1}\ Q_{2}\ \cdots\ Q_{n}]$  be a unitary transformation where $Q_1=x_0$ or equivalently  $x_0=Q e_1$ where $e_1$ is the canonical  vector with 1 as the first entry and zero elsewhere.  Since unitary transformations do not affect the covariance structure of Gaussian random vectors, the matrix $A^*Q$
is an i.i.d. complex standard Gaussian matrix. 
\commentout{In this representation, it follows from \eqref{52}-\eqref{571} that 
\[
\beta=x_{\rm null}(1),\quad z(1)=\sqrt{1-\beta^2},\quad x_\perp(1)=0. 
\]
 }

\begin{prop} Let $I$ 
be any set such that $b(i)\leq b(j)$ for all $i\in I$ and $j\in I_{\rm c} = \{1,2,...,N\}\setminus  I$.
For any unitary matrix $Q$,  let  $A'\in \IC^{|I|\times (n-1)}$ be the sub-column matrix of $A_I^*Q$ with its first column vector deleted.
 Then $A'$ is an  i.i.d. complex  standard Gaussian random matrix. \end{prop}

\begin{proof}
First note that $A_I^*Q= (A^*Q)_I$, the row submatrix of $A^*Q$ indexed by $I$.
As noted already, $A^*Q$ is an i.i.d. complex Gaussian matrix.

Since $x_0=Qe_1$
and $b=|A^*Qe_1|$, $I$ and $I_c$ are entirely determined by the first column of $A^*Q$
which is independent of the other columns of $A^*Q$. 
Consequently, the probability law of $A'$ conditioned on the choice of $I$ equals
the probability law of $A'$ for a fixed $I$. 
Therefore, $A'$ is an i.i.d. complex standard Gaussian matrix. 
\end{proof}

Let $\{\nu_{i}\}_{i=1}^{n-1}$ be the singular values of $A'$ in the ascending order.
For ant $z\in \IC^{n-1}$,  \[
 B':=A'\,\diag(z/|z|)
 \]  has 
 the same set of singular values as $A'$.  Again, we adopt the convention that
 $z(j)/|z(j)|=1$ when $z(j)=0$. We have
\[
\|A' z \|=\|B'\, | z| \|
\] 
and hence
\[
\|A'z \|=( \| \Re(B') \, |z|  \|^2+\| \Im(B')\,  |z|  \|^2)^{1/2}\ge \sqrt{2}\lt(\| \Re(B') \, |z|  \|\wedge \| \Im(B')\,  |z|  \|\rt).
\]

By the  theory of Wishart matrices \cite{DS},  the singular values $\{\nu_j^R \}_{j=1}^{n-1}, \{\nu_j^I \}_{j=1}^{n-1}$ (in the ascending order)  of  $\Re(B'), \Im(B')$ satisfy  the probability bounds
that for every $t>0$ and $j=1,\cdots,n-1$
 \beq
\label{w1} \mathbb{P}\lt(\sqrt{|I|}-(1+t)\sqrt{n}\le \nu_j^R\le \sqrt{|I|}+(1+t)\sqrt{n}\rt)&\ge &1-2e^{-nt^2/2},\\
 \label{w2} \mathbb{P}\lt(\sqrt{|I|}-(1+t)\sqrt{n}\le \nu_j^I\le \sqrt{|I|}+(1+t)\sqrt{n}\rt)&\ge& 1-2e^{-nt^2/2}.
 \eeq
 By Proposition \ref{prop5.1} and \eqref{w1}-\eqref{w2}, we have
 \begin{eqnarray*}
\|x_0x_0^* -x_{\rm null}x_{\rm null}^*\|&\le &{ \sqrt{2}\|b_I\|\over \| \Re(B') \, |y|  \|\wedge \| \Im(B')\,  |y|  \|}\\
&\le&\sqrt{2} \|b_I\|(\nu^R_{n-1}\wedge \nu^I_{n-1})^{-1}\\
&\le& \sqrt{2} \|b_I\|(\sqrt{|I|}-(1+t)\sqrt{n})^{-1}. 
\eeqn
By Proposition \ref{BnormA}, we obtain the desired bound \eqref{error}. 
The success probability is at least the expression \eqref{prob'} minus $4e^{-nt^2/2}$ which equals the expression  
\eqref{prob}.
 

\subsection{Proof of Proposition~\ref{BnormA} }

By the Gaussian assumption, 
   $b(i)^2=|a_i^* x_0|^2$ has a  chi-squared distribution with  the probability density $e^{-z/2}/2$ on $z\in [0,\infty)$ and the cumulative distribution  \begin{eqnarray*}
 F(\tau):=\int_{0}^\tau  2^{-1}\exp(-z/2) dz=1-\exp(-\tau/2). \end{eqnarray*}
 Let  \beq\label{62'} 
\tau_*=-2\ln (1-|I|/N)
\eeq
for which $ F(\tau_*)=|I|/N.$
 
Define
\[ \hat I:=\{i:   b(i)^2\le  \tau_*  \}=\{ i: F(b^2(i))\le |I|/N \},\] and 
\[
\|\hat b\|^2:=\sum_{i\in \hat I} b(i)^2.
\]

Let  \[ \{\tau_1\le \tau_2\le \ldots\le \tau_N\}\] be the sorted sequence of  $\{b(1)^2,\ldots, b(N)^2\}$ in magnitude. 
  \begin{prop}\label{NN} 
{\bf (i)} For any $\delta>0$, we have
   \beq
 \tau_{|I|}&\le& \tau_*+\delta
 \eeq
 with probability at least 
 \beq\label{A3}
 1-\exp\left(-\frac{N}{2} {\delta^2 e^{-\delta} |1-{{|I|/N}}|^2} \right)
   \eeq

{\bf (ii)} For each $\epsilon>0$,
we have
  \beq\label{ZZ} |\hat I| \ge |I| (1-\epsilon)\eeq
  or 
 equivalently,   \beq \label{69} \tau_{\lfloor |I| (1-\epsilon)\rfloor }\le \tau_*
  \eeq
   with probability at least
      \beq
   \label{69.5}
 1-  2\exp\lt(-{4\ep^2 |1-{|I|/N}|^2} |I|^2/N\rt) 
\eeq
   \end{prop}
  \begin{proof} 
  
  {\bf (i)}
 Since  $F'(\tau)= \exp(-\tau/2)/2$, 
 \beq |F(\tau+\epsilon)-F(\tau)|\ge \epsilon/2 \exp(-(\tau+\epsilon)/2).\label{a.16}\eeq
 For $\delta>0$, let
 \[\zeta:=  F(\tau_*+\delta)-F(\tau_*)\]
 which by \eqref{a.16} satisfies
 \beq
 \label{70}
\zeta \ge \frac{\delta}{2}  \exp(-\frac{1}{2}( \tau_*+\delta)). 
\eeq

Let   $\{w_i: i=1,\ldots, N\}$ be  the i.i.d. indicator random variables
   \[
   w_i=\chi_{\{b(i)^2>\tau_*+\delta\}}
   \]
   whose expectation is given by
   \[ \IE[w_i]=1-F(\tau_*+\delta).
   \] 
    The  Hoeffding inequality  yields
     \begin{eqnarray}
\label{ineq1}
 \mathbb{P}(\tau_{|I|} >\tau_*+\delta)
&=&\mathbb{P}\left(\sum_{i=1}^N w_i >N-|I|\right)\\
&=&\mathbb{P}\left(N^{-1}\sum_{i=1}^N w_i-\IE[w_i] >1-|I|/N-\IE[w_i]\right)\nn\\
&=&\mathbb{P}\left(N^{-1}\sum_{i=1}^N w_i-\IE[w_i] >\zeta\right)\nn\\
&\le& \exp(-2N\zeta^2).\nn\end{eqnarray}
Hence, for any fixed $\delta>0$, 
 \beq
  \tau_{|I|}\le& \tau_*+\delta
 \eeq
 holds
 with probability at least 
  \beqn
  \label{a.19}
 1-\exp(-2N\zeta^2)
 &\ge &1-\exp\left(-\frac{N\delta^2}{2} e^{-\tau_*-\delta}\right)\\
  &= &1-\exp\left(-\frac{N \delta^2}{2} e^{-\delta} \left |1-{|I|/N} \right|^2 \right) \nn
   \eeqn
 by \eqref{70}.

{\bf (ii)}    Consider the following replacement
\[
\begin{array}{cl}
(a)   & |I| \longrightarrow \lceil |I| (1-\epsilon)\rceil    \\
(b)   & \tau_* \longrightarrow F^{-1}(\lceil |I| (1-\epsilon)\rceil /N)   \\
(c)   & \delta \longrightarrow F^{-1}(|I|/N)-F^{-1}(\lceil |I| (1-\epsilon)\rceil /N)   \\
(d)   & \zeta \longrightarrow F^{-1}(\tau_*+\delta)-F^{-1}(\tau_*)=|I|/N-\lceil |I| (1-\epsilon)\rceil /N= \frac{\lfloor |I|\epsilon\rfloor}{N}\end{array}
\]
in the preceding argument.
Then (\ref{ineq1}) becomes
 \beqn
\IP\lt(\tau_{\lceil |I|(1-\epsilon)\rceil}>F^{-1}(|I|/N) \rt)& \le & \exp(-2N\zeta^2)
 = \exp\left(-\frac{ 2\lfloor |I|\epsilon\rfloor^2}{N} \right).
 \eeqn
 That is, \[
\tau_{\lceil |I|(1-\epsilon)\rceil}\le \tau_* \]
holds with probability at least
 \beqn
1-\exp(-2 {\lfloor  |I|\epsilon \rfloor^2 }/{N}).
 \eeqn

\end{proof}

   \begin{prop}\label{bhatb}
For each $\epsilon>0$ and $\delta>0$,
 \beq\label{66}
\frac{\|b_I\|^2}{|I|}\le \frac{ \|\hat b\|^2}{|\hat I|}+
\epsilon  (\tau_*+\delta)\eeq
with  probability at least 
\beq \label{77}
 1-2\exp\left(-\frac{1}{2} {\delta^2 e^{-\delta} |1-{|I|/N}|^2 N} \right)-
 2\exp\lt(-{2\ep^2 |1-{|I|/N}|^2}{|I|^2\over N}\rt). 
 \eeq
 
 \end{prop}
 \begin{proof}
Since $\{\tau_j\}$ is an increasing sequence, the function  $T(m)=m^{-1}\sum_{i=1}^m\tau_i$  is also increasing. 
Consider the two alternatives
either   $|I|\ge |\hat I|$ or $|\hat I|\ge |I| $. 
For the latter,   
\[
{\|b_I\|^2}/{|I|}\le { \|\hat b\|^2}/{|\hat I|} 
\]
 due to the monotonicity of $T$.

 For the former case  $|I|\ge |\hat I|$,  we have
 \beqn
T(|I|)&= &|I|^{-1}\sum_{i=1}^{|\hat I|} \tau_i+|I|^{-1}\sum_{i=|\hat I|+1}^{|I|} \tau_i\\
 &\le &T(|\hat I|)+|I|^{-1} (|I|-|\hat I|) \tau_{|I|}.
 \eeqn
By  Proposition \ref{NN}  (ii)  $|\hat I|\geq (1-\ep) |I|$ and hence
 \beqn
T(|I|)&\le &T(|\hat I|)+|I|^{-1}(|I|-|I|(1-\epsilon))\tau_{|I|}
=T(|\hat I|)+\epsilon\tau_{|I|}
\eeqn
 with probability at least given by \eqref{69.5}.

By Proposition \ref{NN} (i),  
$
\tau_{|I|}\le \tau_* +\delta
$
with probability at least given by \eqref{A3}.
\end{proof}

Continuing the proof of Proposition \ref{BnormA},  let us
 consider the i.i.d. centered, bounded random variables \beq
 \label{78}
 Z_i := 
{N^2\over |I|^2} \lt[b(i)^2\chi_{\tau_*}-\IE[b(i)^2\chi_{\tau_*}]\rt] 
 \eeq
 where $\chi_{\tau_*}$ is the characteristic function of the set $\{b(i)^2\leq \tau_*\}$. 
Note that  \beq
\IE(b(j)^2\chi_{\tau_*})&= &\int_{0 }^{\tau_*}  2^{-1} z\exp(-z/2) dz=2-(\tau_*+2)\exp(-\tau_*/2)\le 2 |I|^2/N^2\label{sigmaB} \eeq
and
hence
\beq\label{a.23}
-2 \le Z_i \le \sup\lt\{{N^2\over |I|^2} b(i)^2\chi_{\tau_*} \rt\}= {N^2\over |I|^2} \tau_*.
\eeq

 \commentout{
By the Hoeffding inequality, we then have
  \beq
   \label{a.23'}
   \mathbb{P}\{|\sum_{i=1}^N Z_i|\ge  N t \}\le 2 \exp\lt(-{2N^2 t^2\over N (2+\tau_* N^2/|I|^2)^2}\rt)
   \eeq
}

Next recall the Bernstein-inequality.
   \begin{prop}\cite{Roman} Let $Z_1,\ldots,Z_N$ be i.i.d.  centered sub-exponential
    random variables. Then for every $t\ge 0$, we have
   \beq
   \label{a.21}
   \mathbb{P}\lt\{ N^{-1}|\sum_{i=1}^N Z_i|\ge  t \rt\}\le 2 \exp\lt\{-c \min(Nt^2/K^2, Nt/K)\rt\},
   \eeq
   where $c$ is an absolute constant and 
   \[
   K=\sup_{p\ge 1} p^{-1} (\IE|Z_j|^p)^{1/p}.
   \]
        \end{prop}
        
        \begin{rmk}
   For $K$ we have the following estimates
   \beq\label{a.22}
K&\le& {2N^2\over |I|^2} \sup_{p\ge 1} p^{-1} (\IE| b(i)^2\chi_{\tau_*} |^p)^{1/p} \\
& \le & {2N^2\over |I|^2}  
\tau_* \sup_{p\ge 1} p^{-1}( \IE\chi_{\tau_*})^{1/p}\nn\\
&\le& {2N^2\over |I|^2}  
\tau_* \sup_{p\ge 1} p^{-1}(1-e^{-\tau_*/2})^{1/p}.\nn
\eeq
The maximum of the right hand side of \eqref{a.22} occurs at
\[
p_*=-\ln (1-e^{-\tau_*/2})
\]
and hence
\beqn
K&\le &{2N^2\over |I|^2} {\tau_* \over p_*} (1-e^{-\tau_*/2})^{1/p_*}.
\eeqn
We are interested in the regime
\[
\tau_*\asymp 2|I|/N \ll 1
\]
which implies 
\[
p_*\asymp -\ln {\tau_*\over 2}\asymp \ln{N\over |I|}
\]
and consequently 
\beq
\label{a.24}
K\le {4N\over e |I|} \lt(\ln {N\over |I|}\rt)^{-1},\quad \sigma=|I|/N\ll 1.
\eeq

On the other hand, upon substituting the asymptotic bound \eqref{a.24} in  the probability bound \[
Q=2 \exp\lt\{-c \min(Nt^2/K^2, Nt/K)\rt\}
\]
of \eqref{a.21}, we have 
\[
 K\le 2 \exp\lt\{-c\min \lt[{e^2t^2 \over 16} \lt(\ln \sigma^{-1}\rt)^2 {|I|^2/N},~{et\over 4}|I|\ln\sigma^{-1}\rt]\rt\}, \quad\sigma \ll 1.\\
 \]
     \end{rmk}
        
      The Bernstein  inequality
     ensures  that
  with  high probability
 \[
\left |{ \|\hat b\|^2\over N}- \IE(b^2(i)\chi_{\tau_*})\right|\le t{|I|^2\over N^2}.
 \]
By (\ref{ZZ}) and \eqref{sigmaB}, we also have
   \beq\label{80}
 \frac{\|\hat b\|^2}{|\hat I|}&\le  & \IE(b(i)^2\chi_{\tau_*}) \frac{N}{ |\hat I|}+t\frac{|I|^2}{|\hat I| N}\\
& \le&
\left(\IE(b(i)^2\chi_{\tau_*}) \frac{N^2}{|I|^2}+t\right) \frac{|I|}{N}\nn\\
&\le& \frac{2+t}{1-\epsilon} \cdot\frac{|I|}{N}\nn 
 \eeq

By Prop.~\ref{bhatb}, we now have
\beqn
\|b_I\|^2&\le &|I| \lt({\|\hat b\|^2\over |\hat I|} +\ep \lt(\tau_*+\delta\rt)\rt)
 \eeqn
with probability at least given
by \eqref{prob},
 which together with \eqref{80} and \eqref{62'} complete the proof of Proposition~\ref{BnormA}. \\

\commentout{

\section{The spectral initialization}\label{sec:spectral}
Here we compare the null initialization with the spectral initialization used in \cite{CLS2} and the truncated spectral initialization used in \cite{truncatedWF}.

\begin{algorithm}[h]
\SetKwFunction{Round}{Round}
\textbf{Random initialization:} $x_1=x_{\rm rand}$\\
\textbf{Loop:}\\
\For{$k=1:k_{\textup{max}}-1$}
{
$x_k'\leftarrow A(|b|^2\odot A^*x_k);$\\
$x_{k+1}\leftarrow [x_k^{'}]_\cX/\|[x_k^{'}]_\cX\|$;
}
{\bf Output:} $x_{\rm spec}=x_{k_{\rm max}}$.
\caption{\textbf{The spectral initialization}}
\label{spectral-algorithm}
\end{algorithm}

The key difference between Algorithms 1 and 2 is the different weights used in
step 4 where  the null initialization uses $\mathbf{1}_{c}$ and  the spectral vector
method uses $|b|^2$ (Algorithm 2). 
The truncated spectral initialization uses  a still  different weighting 
\begin{equation}\label{truncated_version}
x_{\textup{t-spec}}=\textup{arg}\underset{\|x\|=1}{\textup{max}}
\| A
\left(\mathbf{1}_\tau\odot |b|^{2}\odot A^*x \right)\|
\end{equation}
where 
$\mathbf{1}_\tau$ is the characteristic function of 
the set
\[
\{i: |A^* x(i)|\leq \tau {\|b\|}\}
\]
with an adjustable parameter $\tau$. Both $\gamma$ of Algorithm 1 and $\tau$ of \eqref{truncated_version} can be 
optimized  by tracking and minimizing  the residual $\|b-|A^* x_k|\|$.

As shown in the numerical experiments in Section \ref{sec:num} (Fig. \ref{fig:initials} and~\ref{fig:initials_2masks}), the choice of weight  significantly
affects  the quality of initialization, with the null initialization as  the best performer
(cf.  Remark \ref{rmk5.2}). 

Moreover, because the null initialization depends only on the choice of the index set $I$ and not explicitly on $b$, the method is noise-tolerant and performs well with noisy data (Fig. \ref{fig:noise0}). 
}

\commentout{
\subsection{Noise stability}

We test the performance of AP and WF with the Gaussian noise model
where the noisy data is generated by 
\[
b_{\rm noisy}=|A^{*}x_{0}+ \hbox{complex\ Gaussian\ noise} |.
\]
The noise  is measured by the Noise-to-Signal Ratio (NSR)
\[
\textup{NSR}=\frac{\|b_{\rm noisy}-|A^* x_0|\|_2}{\|A^* x_0\|_2}.
\]

As  pointed out in Section \ref{sec:spectral},  
since  the null initialization depends only on the choice of the index set $I$  and does not depend explicitly on $b$, the method  is more noise-tolerant than other initialization methods.

Let $\hat x_{\rm null}$ be the {\em unit} leading singular vector of $A_{I_c}$, cf. \eqref{2.4}. In order to compare the effect of normalization, we  normalize the null vector  in two different ways 
\beq
\hbox{\rm Case 1.}\quad x_{\rm null}&=&\alpha{\|b_{\rm noisy}\|}\cdot\hat x_{\rm dual}\\
\hbox{\rm Case 2.}\quad x_{\rm null}&=&\alpha {\|x_0\|}\cdot\hat x_{\rm dual}
\eeq
and then compute their respective relative errors versus NSR.

If $\|x_0\|$ is known explicitly, we can apply AP with the normalized
noisy data
\[
\hat b_{\rm noisy}=b_{\rm noisy} {\|x_0\|\over \|b_{\rm noisy}\|}
\]
and improve the performance. 

\section{Conclusion and discussion}\label{sec:sum}
}

{\bf Acknowledgements.} Research of P. Chen is supported in part by the grant 103-2115-M-005-006-MY2 from Ministry of Science and Technology, Taiwan,  and US NIH grant U01-HL-114494. Research of A. Fannjiang  is supported in part by  US National Science Foundation  grant DMS-1413373 and Simons Foundation grant 275037.

\end{document}